\documentclass{eptcs}
\providecommand{\event}{TARK 2017} 

\newcommand{\commentout}[1]{}

\usepackage{booktabs} 

\commentout{
\usepackage[ruled]{algorithm2e} 

\SetAlFnt{\small}
\SetAlCapFnt{\small}
\SetAlCapNameFnt{\small}
\SetAlCapHSkip{0pt}
\IncMargin{-\parindent}
}

\usepackage{amsfonts}
\usepackage{mathrsfs}
\usepackage{amssymb}
\usepackage{latexsym}
\usepackage{graphicx}
\usepackage{enumerate}
\usepackage{amsmath}
\usepackage{color}
\usepackage{float}
\usepackage{array}
\usepackage{mathtools}
\mathtoolsset{centercolon}

\usepackage{amsthm}
\newtheorem{theorem}{Theorem}
\newtheorem{lemma}[theorem]{Lemma}

\newtheorem{proposition}[theorem]{Proposition}

\usepackage{amssymb,tikz}

\newcommand{\mysetminusD}{\hbox{\tikz{\draw[line width=0.6pt,line cap=round] (3pt,0) -- (0,6pt);}}}
\newcommand{\mysetminusT}{\mysetminusD}
\newcommand{\mysetminusS}{\hbox{\tikz{\draw[line width=0.45pt,line cap=round] (2pt,0) -- (0,4pt);}}}
\newcommand{\mysetminusSS}{\hbox{\tikz{\draw[line width=0.4pt,line cap=round] (1.5pt,0) -- (0,3pt);}}}

\newcommand{\mysetminus}{\mathbin{\mathchoice{\mysetminusD}{\mysetminusT}{\mysetminusS}{\mysetminusSS}}}

\newcommand{\defin}[1]{\textbf{#1}}

\newcommand{\lthen}{\rightarrow}

\newcommand{\defeq}{\coloneqq}

\newcommand{\val}[1]{[\![ #1 ]\!]}

\newcommand{\B}{\mathcal{B}}
\newcommand{\T}{\mathcal{T}}
\newcommand{\F}{\mathcal{F}}
\newcommand{\M}{\mathcal{M}}
\newcommand{\I}{\mathcal{I}}
\renewcommand{\P}{\mathcal{P}}
\newcommand{\A}{\mathcal{A}}
\renewcommand{\phi}{\varphi}
\newcommand{\citeyear}{\cite}

\newcommand{\rimp}{\Rightarrow}
\newcommand{\dimp}{\Leftrightarrow}

\newcommand{\PR}{\mathit{Pr}}
\newcommand{\FT}{\F_{\T}}
\newcommand{\LBT}{\mathcal{L}_{B}^{\Theta}}
\newcommand{\LBI}{\mathcal{L}_{B}^{[0,1]}}
\newcommand{\LBII}{\mathcal{L}_{B}^{\infty,[0,1]}}
\newcommand{\valM}[1]{\val{#1}_{\M}}
\newcommand{\valI}[1]{\val{#1}_{\I}}
\newcommand{\MI}{\M_{\I}}
\newcommand{\TM}{\T_{\M}}
\newcommand{\IM}{\I_{\M}}

\begin{document}

\title{From Type Spaces to Probability Frames and Back, via Language} 
\author{Adam Bjorndahl
  \institute{Carnegie Mellon University\\
  Department of Philosophy\\
  Pittsburgh, PA 23185, USA}
  \email{adam.bjorndahl@gmail.com}
\and
  Joseph Y. Halpern
\institute{Cornell University\\
Ithaca, NY 14853, USA}
  \email{halpern@cs.cornell.edu}
}

\def\titlerunning{From Type Spaces to Probability Frames and Back, via
  Language} \def\authorrunning{A. Bjorndahl \& J. Y. Halpern}

\maketitle

\begin{abstract}
We investigate the connection between the two major mathematical
frameworks for modeling interactive beliefs: \emph{Harsanyi type
  spaces} and possible-worlds--style \emph{probability frames}. While
  translating the former into the latter is straightforward, we
demonstrate that the reverse translation relies implicitly on a
background logical language. Once this ``language parameter'' is made
explicit, it reveals a close relationship between \emph{universal}
type spaces and \emph{canonical} models: namely, that they are
essentially the same construct. As the nature of a canonical model
depends heavily on the background logic used to generate it, this work
suggests a new view into a corresponding landscape of universal type
spaces. 
\end{abstract}

%
%




\section{Introduction}

There are two mathematical frameworks in widespread use for modeling
beliefs in multi-agent systems. One approach, popular among computer
scientists and logicians, utilizes the \emph{possible worlds}
paradigm (see, e.g., \cite{Hal31}). Roughly
speaking, a \emph{probability frame} consists of a set of
\emph{worlds}, each of which is associated with a set of probability
measures (one for each agent), defined on the set of worlds. These
probability measures are interpreted as encoding beliefs. Hierarchical
beliefs---for example, beliefs about what another agent believes---are
naturally captured by the recursive structure of this framework,
namely the fact that worlds encode beliefs about worlds. 
The second approach, more standard in game theory, uses
\emph{type spaces}, introduced by Harsanyi \citeyear{Harsanyi}.
Roughly speaking, types spaces are composed of
\emph{states}, encoding ``basic'' facts about the world (typically
including which strategies the players are using), together with
\emph{types}, encoding the beliefs of each player in the form of a
probability measure defined over the states and the types of her
opponents. 

What is the relationship between probability frames and type spaces?
Aside from a few measure-theoretic technicalities, it is relatively
straightforward to transform a type space into a probability frame:
essentially, the worlds are state-type pairs. 
Reversing this transformation is not so straightforward. Given a
probability frame,
the key question is how to ``factor'' worlds into states and types.
Probability frames encode beliefs about worlds, beliefs about beliefs
about worlds, and so on, but this never ``bottoms out'' in anything
like the states in a type space.  
That is, there is no obvious component of a world that encodes facts 
such as what strategies the agents are using or the value participants in
an auction might assign to an item up for bid.
Thus, there seems to be a mismatch between the two approaches.

In this paper, we resolve this mismatch by adding a \emph{language}---a
set of basic facts (such as what strategy is used by each agent),
represented by primitive propositions---to
the picture. In the terminology of modal logic, we
pass from frames to models.
Given a language, a model is simply a frame together with an
\emph{interpretation} that determines for each world $w$ and primitive
proposition $p$ in the language whether $p$ is true in world $w$.
But then we must decide which language to use.
We show that the right choice of language can provide exactly the
additional structure needed to ``cleanly'' factor worlds into states
and types. Specifically, we define a transformation on probability
models that takes language as a parameter, and show that it produces
the familiar type space construction when the language is
appropriately expressive. 

The value of forging such a connection between the
two major mathematical frameworks for modeling belief is obvious:
improved communication between researchers working in these respective
traditions, and the prospect of importing insights and results from
one paradigm to another.
And indeed, one immediate application of our language-sensitive
translation is a link between two fundamental notions: that of a
\emph{canonical} model from the world of modal logic (see, e.g.,
\cite{BRV01}) and that of a \emph{universal} type space from the
theoretical economics literature \cite{MZ}. 
Each of these constructions plays a central role in the subfield to
which it belongs, and these roles are very similar: each is, in a
precise sense, the ``largest'' structure of its kind---a structure
that essentially contains all other such structures. It is perhaps not
surprising that they are effectively the \textit{same} structure:
roughly speaking, we show in Section \ref{sec:trs} that canonical
models are transformed into universal type spaces.%
\footnote{We remark that Meier \citeyear{Meier12} already observed
  this connection in the case of an infinitary language.}
Moreover, since
canonical model constructions are highly sensitive to the underlying
logical language, this result suggests a new view into a landscape of
universal type spaces parametrized by language. 

Much of this work was inspired by a beautiful paper of Heifetz and
Samet \citeyear{HeSa98}. 
In it, they construct a measure-theoretic universal type space by a
process that closely mimics a standard canonical-model construction
(though they do not describe it that way). Our work can be viewed as
generalizing their construction to produce a translation from
\textit{arbitrary}
probability frames to type spaces; our Theorem \ref{thm:can} is
then the special case of applying this translation to the canonical model
associated with a certain specific logic. In order to emphasize this
connection, much of the notation and terminology of this paper
duplicates or parallels that used by Heifetz and Samet.

In fact, our ``canonical model'' construction differs
in small but significant ways
from the standard construction in modal logic. Typically,
worlds in the canonical model are realized as maximal
\textit{consistent} sets of formulas from the language, where
consistency is, of course, defined relative to some background axiom
system. However, the standard finitary axiom system used to reason
about probability frames has a
problem, namely, it is not
\emph{compact}: there exists an infinite set $F$ of formulas that is 
not satisfiable such that every finite subset of $F$ is satisfiable
(which means that $F$ is consistent with the axioms).
This renders the
corresponding canonical model not a model at all. To avoid this issue,
we replace ``consistency'' with ``satisfiability'' in our canonical
model construction. (Aumann \citeyear{Aumann99a} uses an
analogous construction.)
Meier~\citeyear{Meier12} considers an alternative approach: 
changing the axiom system. Specifically, he considers an
\textit{infinitary} axiom system (with infinitary rules of inference)
with respect to which consistency and satisfiability coincide, and
constructs a universal type space using a canonical model style
construction over this infinitary logic.
Although Meier's logic is infinitary (he allows uncountable
conjunctions and disjunctions) and our language is finitary, his
canonical model is essentially isomorphic to ours
(see Section~\ref{sec:uni} for further discussion).%
\footnote{We thank Martin Meier for pointing this important connection between
  our work, his work, and that of Aumann.}
Conceptually, however, our goals are somewhat different from
those of Aumann and Meier.
Aumann and Meier focus on the construction of the canonical model.  By
way of contrast, 
we approach the issue as a problem first of how
to transform an arbitrary probability frame into a type space, and
observe afterwards that this translation connects a (suitably defined)
notion of canonical model to that of a universal type space.

We are not the first to study the general relationship between type
spaces and
possible-worlds--style structures.
One connection via logic is well known.  Sound and complete
axiomatizations have been provided for various logics of probability:
Heifetz and Mongin \cite{HeifetzM01} considered a finitary logic where
the basic statements have the form $B_i^\theta 
\phi$ (agent $i$ believes that the probability of $\phi$ is at least
$\theta$)---this is the same logic that we consider---and provided a
sound and complete axiomatization in their logic for type systems; Meier
\citeyear{Meier12} did the same for an infinitary logic.  Since the
axioms are easily seen to be sound in probability frames, and 
every type structure can be viewed as a probability frame,
soundness and completeness of these axiomatizations for
probability frames follows.  Fagin, Halpern, and Megiddo
\cite{FHM,FH3} provided a sound and complete axiomatization of a
logic that allowed reasoning about linear combinations
of probabilities (i.e., statements such as $2\ell_i(\phi) + 3 \ell_i(\psi)
\ge 1.5$, which can be read as ``twice agent $i$'s probability of $\phi$
plus three times agent $i$'s probability of $\psi$ is at least 1.5'') in
probability frames.  Since their axioms are easily seen to be
sound in type spaces and statements about linear combinations can be
expresssed in Meier's infinitary logic, it follows that this
axiomatization is also complete for type spaces.

The work on axiomatizations does not
produce an explicit translation
between type spaces and possible-worlds structures.  In more recent
work, Galeazzi and Lorini \citeyear{GL16}
develop a translation between the two and prove a semantic equivalence
result. They, too, work at the level of models rather than frames
(though they do not explicitly discuss this choice); however, their
translations are defined model-theoretically with respect to a single
fixed language, rather than taking language as a parameter, making the
approach we develop more flexible and more broadly applicable.
While the translation they propose from (what we call) probability
models into type spaces is not a special case of
ours, it is similar in spirit. However, there is one significant
difference: in passing through language, our approach 
effectively identifies worlds that satisfy all the same formulas,
while theirs does not (in particular, ``duplicate'' worlds produce
duplicate types under their translation, but not under
ours). Semantically speaking, provided we fix an
appropriately expressive language, the type spaces we produce are
equivalent, once  we identify types that satisfy the same formulas.
By varying the language, however, our translations take on
different characters---they preserve more or less of the type space
structure in accordance with what is expressible in the 
language. 
Moreover, Galeazzi and Lorini restrict their attention to
countable structures, which effectively precludes
consideration of
structures like universal type spaces or canonical models.

The rest of the paper is organized as follows. Section \ref{sec:pre}
presents the basic mathematical frameworks within which we
work. Section \ref{sec:tra} motivates and defines the translations
from type spaces to probability frames and vice-versa. Section
\ref{sec:uni} presents the connection between universal type spaces
and canonical models discussed above.
Section \ref{sec:con} concludes.
Some proofs have been omitted or abridged due to length requirements.

\section{Preliminaries} \label{sec:pre}

The definition of a type space typically includes various topological
assumptions that make it easier to prove certain results of interest
within that framework \cite{DS15}.
Since our goal is to understand the connection between type spaces and
probability frames, we opt instead to work in as minimal a setting as
possible, so as not to obscure the translations between the two with
additional topological bookkeeping. In particular, following Heifetz
and Samet \citeyear{HeSa98},
we work with a purely measure-theoretic definition of types spaces.

A \defin{measurable space} is a set $X$ together with a $\sigma$-algebra $\Sigma_{X}$ over $X$; elements of $\Sigma_{X}$ are called \defin{measurable sets} or \defin{events}. We often drop explicit mention of $\Sigma_{X}$ and refer simply to ``the measurable space $X$''. We denote by $\Delta(X)$ the measurable space of all probability measures on $X$ equipped with the $\sigma$-algebra generated by all sets of the form
$$\B^{\theta}(E) \defeq \{\mu \in \Delta(X) \: : \: \mu(E) \geq \theta\},$$
where $\theta \in [0,1]$ and $E \in \Sigma_{X}$ is an event. Given measurable spaces $X_{1}, \ldots, X_{k}$, the measurable space $X_{1} \times \cdots \times X_{k}$ is just the usual product space equipped with the $\sigma$-algebra generated by all sets of the form $E_{1} \times \cdots \times E_{k}$, where each $E_{i} \in \Sigma_{X_{i}}$.

Given a probability measure $\mu$ on $X$, the associated \defin{outer measure}, denoted $\mu^{*}$, is defined on arbitrary subsets of $X$ as follows:
$$\mu^{*}(A) \defeq \inf\{\mu(E) \: : \: E \in \Sigma_{X} \textrm{ and } E \supseteq A\}.$$
Obviously, if $A \in \Sigma_{X}$, then $\mu^{*}(A) =
\mu(A)$. Otherwise, if $A$ is not a measurable set, the outer measure
of $A$ can be thought of as a kind of approximation of
the measure of $A$ from above: every event containing $A$ has
probability at least
$\mu^{*}(A)$, and for all $\varepsilon > 0$, there is an event $E
\supseteq A$ with $\mu(E) - \mu^{*}(A) < \varepsilon$. 

Fix a finite set $I = \{1, \ldots, n\}$ of agents. We adopt the usual
notational game-theoretic conventions for tuples over $I$: Given
$(X_{i})_{i \in I}$, we write 
$$X \defeq \prod_{i \in I} X_{i} \quad \textrm{and} \quad X_{-i} \defeq \prod_{j \neq i} X_{j}.$$
We also write $X_{i}' \times X_{-i}$ for
$$X_{1} \times \cdots \times X_{i-1} \times X_{i}' \times X_{i+1} \times \cdots \times X_{n}$$
and similarly $(x_{i}', x_{-i})$ for
$$(x_{1}, \ldots, x_{i-1}, x_{i}', x_{i+1}, \ldots, x_{n}).$$

A \defin{type space (over $I$)} is a tuple $\T = (X, (T_{i})_{i \in I}, (\beta_{i})_{i \in I})$ where
\begin{itemize}
\item
$X$ is a measurable space of \emph{states};
\item
$T_{i}$ is a measurable space of \emph{$i$-types};
\item
$\beta_{i}: T_{i} \to \Delta(X \times T)$ is a measurable function such that the marginal of $\beta_{i}(t_{i})$ on $T_{i}$ is $\delta_{t_{i}}$, the point-mass measure concentrated on $t_{i}$.
\end{itemize}
Intuitively, $X$ captures the basic facts about which the agents may be uncertain, while $i$-types represent the beliefs of agent $i$ via the function $\beta_{i}$. These beliefs are not just about the states, but also about the types (and therefore the beliefs) of the agents. In this context, the requirement that $\beta_{i}$ be measurable can be thought of as a closure condition on events: for all events $E \subseteq X \times T$, the set of points where agent $i$ assigns $E$ probability at least $\theta$, namely
$$X \times \beta_{i}^{-1}(\B^{\theta}(E)) \times T_{-i},$$
is itself an event. The extra condition on $\beta_{i}$ is meant to ensure that agent $i$ is \emph{introspective}: that is, sure of her own beliefs. The point-mass measure $\delta_{t_{i}}$ is defined on the measurable subsets of $T_{i}$ by
$$
\delta_{t_{i}}(E) = \left\{ \begin{array}{ll}
1 & \textrm{if $t_{i} \in E$}\\
0 & \textrm{if $t_{i} \notin E$.}
\end{array} \right.
$$
Thus, $\delta_{t_{i}}$ assigns probability 1 to all and only the
events containing $t_{i}$. Note that in general we cannot simply say
that $\{t_{i}\}$ has probability 1 according to agent $i$, since
$\{t_{i}\}$ may not be measurable; instead, we can say that every
event incompatible with $t_{i}$ has probability 0 according to agent
$i$.\footnote{This subtlety does not typically arise in the richer
  topological setting: provided $T_{i}$ is a $T_{1}$-space
(see, e.g., \cite{Munkres};
  there is
  an unfortunate clash of notation here), $\{t_{i}\}$ is closed and
  therefore part of the Borel $\sigma$-algebra associated with
  $T_{i}$.} Equivalently, $\delta_{t_{i}}$ is the unique probability
measure on $T_{i}$ that assigns $\{t_{i}\}$ outer measure $1$.  
A \defin{probability frame (over $I$)} is a tuple $\F = (\Omega, (\PR_{i})_{i \in I})$ where
\begin{itemize}
\item
$\Omega$ is a measurable space of \emph{worlds};
\item
$\PR_{i}: \Omega \to \Delta(\Omega)$ is a measurable function such that, for each $\omega \in \Omega$, $\PR_{i}(\omega)^{*}(\PR_{i}^{-1}(\PR_{i}(\omega))) = 1$.
\end{itemize}
Here, all information is encoded in $\Omega$, basic facts and beliefs alike. As with type spaces, the measurability of $\PR_{i}$ yields a closure condition on events: for all events $E \subseteq \Omega$, the set of points where agent $i$ assigns $E$ probability at least $\theta$ is given by $\PR_{i}^{-1}(\B^{\theta}(E))$ and is therefore measurable. And as above, the additional condition on $\PR_{i}$ amounts to the stipulation that agent $i$ is sure of her own beliefs in the sense that at each world $\omega$, $\PR_{i}(\omega)$ assigns outer measure $1$ to the set
$$\PR_{i}^{-1}(\PR_{i}(\omega)) = \{\omega' \: : \: \PR_{i}(\omega') = \PR_{i}(\omega)\},$$
namely, the set of worlds where her beliefs are given by the measure
$\PR_{i}(\omega)$. If this set is measurable, of course, then it is
itself assigned probability $1$. In much of the literature
the measurability of this set is simply assumed. We adopt the slightly
more cumbersome definition given above using outer measure because it
is more general and because it parallels the introspection
condition assumed in type spaces in a way that helps to streamline the
translation between the two. 

\section{Translations} \label{sec:tra}

Informally, a type space looks like a probability frame where the set of worlds $\Omega$ has been ``factored'' into a component representing basic facts---the states---and components representing the beliefs of the agents---the types. As discussed in the introduction,
given a probability frame, it is not clear how to perform such a
factorization; most of this section is concerned with developing a
solution to this problem. The reverse construction, on the other hand,
is straightforward, so we begin with it. 

\commentout{
}

\begin{proposition} \label{pro:t-f}
Let $\T = (X, (T_{i})_{i \in I}, (\beta_{i})_{i \in I})$ be a type space, and define $\Omega \defeq X \times T$ and $\PR_{i}(x,t) \defeq \beta_{i}(t_{i})$. Then $\FT \defeq (\Omega, (\PR_{i})_{i \in I})$ is a probability frame.
\end{proposition}

\begin{proof}
This is the obvious construction; all that needs to be checked is that $\PR_{i}$ satisfies the appropriate conditions. Measurability of this function is an easy consequence of the measurability of $\beta_{i}$, since $\PR_{i}^{-1}(\mathcal{E}) = X \times \beta_{i}^{-1}(\mathcal{E}) \times T_{-i}$. For introspection, observe that
$$\PR_{i}(x,t)^{*}(\PR_{i}^{-1}(\PR_{i}(x,t))) =
\beta_{i}(t_{i})^{*}(\{(x',t') \: : \: \beta_{i}(t_{i}') =
\beta_{i}(t_{i})\})
= 1,
$$
since every measurable set containing $\{(x',t') \: : \: \beta_{i}(t_{i}') = \beta_{i}(t_{i})\}$ is of the form $X \times U \times T_{-i}$, where $U \subseteq T_{i}$ is measurable and contains $t_{i}$.
\end{proof}

In what sense is $\FT$ the ``right'' translation of $\T$? Intuitively, we want to say that the relevant properties of agents and their beliefs that are captured by $\T$ are also captured by $\FT$, and in some sense preserved by this translation. To make this precise, we formalize the notion of ``relevant properties'' by identifying them with formulas in a suitably expressive logical language; we then show that the map $\T \mapsto \FT$ is truth-preserving with respect to this language (Proposition \ref{pro:t-m}). In addition to providing a formal standard by which to evaluate purported translations between models, making the background language explicit lays the groundwork for the reverse translation, which makes essential use of this structure.

\subsection{Language} \label{sec:lan}

Fix a set $\Phi$ of \emph{primitive propositions} and a set $\Theta
\subseteq [0,1]$ of \emph{thresholds}; let $\LBT(\Phi, I)$ be the 
language recursively generated by the grammar 
$$\phi ::= p \, | \, \lnot \phi \, | \, \phi \land \psi \, | \, B_{i}^{\theta} \phi,$$
where $p \in \Phi$, $i \in I$, and $\theta \in \Theta$. The parameters
$\Phi$ and $I$ are omitted when they are clear from context. The other
Boolean connectives can be defined in the standard way. We read
$B_{i}^{\theta} \phi$ as ``agent $i$ believes that the probability of
$\phi$ is at least $\theta$''. Intuitively, $\Theta$ collects the set
of thresholds that the language can express beliefs up to. 

There is a standard way of interpreting formulas of $\LBT(\Phi, I)$ in
probability frames. A \defin{probability model (over $(\Phi, I)$)} is
a tuple $\M = (\F, \pi)$ where $\F$ is a probability frame (over $I$)
and $\pi: \Phi \to \Sigma_{\Omega}$ is an \emph{interpretation}. Recall that
$\Sigma_{\Omega}$ denotes the $\sigma$-algebra associated with the
measurable space $\Omega$; the event $\pi(p) \subseteq \Omega$ is
conceptualized as the set of worlds where the primitive proposition
$p$ is true. We can extend this notion of truth to all formulas by
defining $\valM{\cdot}: \LBT \to \Sigma_{\Omega}$ recursively as
follows: 
\begin{eqnarray*}
\valM{p} & = & \pi(p)\\
\valM{\lnot \phi} & = & \Omega \mysetminus \valM{\phi}\\
\valM{\phi \land \psi} & = & \valM{\phi} \cap \valM{\psi}\\
\valM{B_{i}^{\theta} \phi} & = & \{\omega \in \Omega \: : \: \PR_{i}(\omega)(\valM{\phi}) \geq \theta\}.
\end{eqnarray*}
Of course, the final clause of this definition only makes sense if $\valM{\phi}$ is measurable, which follows from an easy induction on formulas using the fact that
$$\valM{B_{i}^{\theta} \phi} = \PR_{i}^{-1}(\B^{\theta}(\valM{\phi})).$$

We say that a formula $\phi$ is \defin{true at $\omega$ (in $\M$)} if
$\omega \in \valM{\phi}$, and that a set $F$ of formulas is true at
$\omega$ if each $\phi \in F$ is true at $\omega$. A formula or set of
formulas is \defin{valid in $\M$} if it is satisfied at all
worlds in $\M$, and \defin{satisfiable in $\M$} if it is true at some
world in $\M$; it is \defin{valid} if it is valid in all
probability models, and \defin{satisfiable} if it is satisfiable in
some probability model. 

It is worth noting that the introspection condition on frames, which says that every event containing $\PR_{i}^{-1}(\PR_{i}(\omega))$ has probability $1$ according to $\PR_{i}(\omega)$, allows us to deduce the following for all probability models $\M$ (assuming $1 \in \Theta$):
\begin{eqnarray*}
\omega \in \valM{B_{i}^{\theta} \phi} & \rimp & \PR_{i}^{-1}(\PR_{i}(\omega)) \subseteq \valM{B_{i}^{\theta} \phi}\\
& \rimp & \PR_{i}(\omega)(\valM{B_{i}^{\theta} \phi}) = 1\\
& \rimp & \omega \in \valM{B_{i}^{1} B_{i}^{\theta} \phi}.
\end{eqnarray*}
This implies that the formula $B_{i}^{\theta} \phi \lthen B_{i}^{1} B_{i}^{\theta} \phi$ is valid: whenever agent $i$ believes the probability of $\phi$ is at least $\theta$, she is sure that she has this belief. A similar argument shows that $\lnot B_{i}^{\theta} \phi \lthen B_{i}^{1} \lnot B_{i}^{\theta} \phi$ is valid. Of course, this also follows from the stronger assumption that $\PR_{i}^{-1}(\PR_{i}(\omega))$ is itself measurable and has probability $1$, but relative to this logical language, such an assumption is overkill.

We can also interpret $\LBT(\Phi, I)$ in type spaces. Although this is
not typically done in the literature (though Galeazzi and Lorini
\citeyear{GL16} do),
it allows us to state formally the connection between $\T$ and $\FT$
as defined in Proposition \ref{pro:t-f}, and it highlights the
analogies between type spaces and probability frames that we exploit
below. 

An \defin{interpreted type space (over $(\Phi, I)$)} is a pair $\I =
(\T, \nu)$ where $\T$ is a type space and $\nu: \Phi \to
\Sigma_{X}$ is an \emph{interpretation}; intuitively, $\nu(p)$
specifies the states of nature where $p$ is true. As above, $\nu$
induces a function $\valI{\cdot}: \LBT \to \Sigma_{X \times T}$ as
follows: 
\begin{eqnarray*}
\valI{p} & = & \nu(p) \times T\\
\valI{\lnot \phi} & = & (X \times T) \mysetminus \valI{\phi}\\
\valI{\phi \land \psi} & = & \valI{\phi} \cap \valI{\psi}\\
\valI{B_{i}^{\theta} \phi} & = & \{(x,t) \in X \times T \: : \: \beta_{i}(t_{i})(\valI{\phi}) \geq \theta\}.
\end{eqnarray*}

Now we can formalize the sense in which the map $\T \mapsto \FT$ is
truth-preserving. 

\begin{proposition} \label{pro:t-m}
Let $\I = (\T, \nu)$ be an interpreted type space, and let $\FT$ be
the probability frame corresponding to $\T$ as defined in Proposition
\ref{pro:t-f}. Define $\pi(p) \defeq \nu(p) \times T$. Then $\MI
\defeq (\FT, \pi)$ is a probability model, and for all $\phi \in
\LBT$, we have $\val{\phi}_{\MI} = \valI{\phi}$. 
\end{proposition}

\begin{proof}
Proposition \ref{pro:t-f} tells us that $\FT$ is a probability frame,
and since $\nu(p) \in \Sigma_{X}$, it is clear that $\pi(p) \in
\Sigma_{X \times T}$; it follows that $\MI$ is a probability model. 

The equality $\val{\phi}_{\MI} = \valI{\phi}$ is proved by an easy structural induction on $\phi$. The base cases where $\phi \in \Phi$ follows from the definition of $\pi$, and the induction steps are all trivial.
\end{proof}

Proposition \ref{pro:t-m} is parametrized by the choice of primitive
propositions $\Phi$ and the interpretation $\nu$: it says that for
any such choice, the correspondence $\T \mapsto \FT$ can be extended
to a correspondence $\I \mapsto \MI$ that is truth preserving with
respect to the language $\LBT(\Phi)$. It is worth emphasizing a
special case of this result. Given a type space
$\T = (X, (T_{i})_{i \in I}, (\beta_{i})_{i \in I})$,
recall that the set $X$ of states
is often conceptualized as representing the ``basic facts'' about the
game; for example, the strategy profiles that may be played. As such,
when $X$ is finite (or even just when $\Sigma_{X}$ contains all
singletons), it is natural to take $\Phi = X$ and define $\nu(x) =
\{x\}$; in this case, intuitively, the primitive propositions simply
say what the true state is.

\subsection{Factoring worlds} \label{sec:fac}

We turn now to the reverse translation: the construction of a suitable
type space from a given probability frame. As we have observed, the
difficulty lies in ``factoring'' worlds into states and types. Given a
probability frame $\F = (\Omega, (\PR_{i})_{i \in I})$, we might hope
to identify types for player $i$ with probability measures of the form
$\PR_{i}(\omega)$ for $\omega \in \Omega$, but what are the states?
This is the crux of the problem: there is nothing in the definition of
$\F$ that allows us to distinguish the ``part'' of a world $\omega$ that represents basic facts; indeed, there is no notion of a
``basic fact'' at all in a probability frame. 

A sufficiently rich logical language, however, such as $\LBT$, \textit{does} distinguish ``basic'' facts from facts about beliefs. For this reason, the construction of a type space naturally operates at the level of probability \textit{models} (which can interpret languages) rather than frames, and depends crucially on the background language.

An \defin{$\LBT$-description} is a set $D \subseteq \LBT$ of formulas 
that is satisfiable and also \emph{maximal} in the sense that, for
each $\phi \in \LBT$, either $\phi \in D$ or $\lnot \phi \in D$. Given
a probability model $\M$ and a world $\omega$ in $\M$, define the
\defin{$\LBT$-description of $\omega$ in $\M$} to be 
$$D(\omega) \defeq \{\phi \in \LBT \: : \: \omega \in \valM{\phi}\}.$$
We omit mention of the language and the model when it is safe to do
so. It is easy to see that $D(\omega)$ is an $\LBT$-description; we
call $D$ the \emph{description map} for $\M$. Intuitively, $D(\omega)$
records all the information about the world $\omega$ expressible in the language
$\LBT$. Let $d_{0}(\omega)$ denote the subset of $D(\omega)$
consisting of the \emph{purely propositional} formulas: that is,
Boolean combinations of the primitive propositions. Let
$d_{i}(\omega)$ consist of the formulas in $D(\omega)$ that are
Boolean combinations of formulas of the form $B_{i}^{\theta}
\phi$. Call these the \defin{$0$-description} and the
\defin{$i$-description} of $\omega$, respectively. We think of the
former as recording the basic facts about $\omega$ (expressible in
$\LBT$), and the latter as recording the beliefs of agent $i$ in
$\omega$ (again, expressible in $\LBT$). 

Fix a probability model $\M = ((\Omega,(\PR_{i})_{i \in I}),\pi)$. We
construct a type space out of $\M$ by 
identifying states with $0$-descriptions and $i$-types with
$i$-descriptions. Formally, set 
$$
\textrm{$X \defeq \{d_{0}(\omega) \: : \: \omega \in \Omega\}$ and $T_{i} \defeq \{d_{i}(\omega) \: : \: \omega \in \Omega\}$.}
$$
Intuitively, each state and each type is constituted by a fragment of information about some world $\omega$ in $\M$. We also use this information to define the measure structure: for each $\phi \in \LBT$, set
$$
\textrm{$E_{0}(\phi) \defeq \{x \in X \: : \: \phi \in x\}$ and $E_{i}(\phi) \defeq \{t_{i} \in T_{i} \: : \: \phi \in t_{i}\}$;}
$$
we consider $X$ and $T_{i}$ as measurable spaces equipped with the $\sigma$-algebras generated by the collections $\{E_{0}(\phi) \: : \: \phi \in \LBT\}$ and $\{E_{i}(\phi) \: : \: \phi \in \LBT\}$, respectively.

The reason we use formulas to pick out events is because, ultimately, we will define each probability measure $\beta_{i}(t_{i})$ on $X \times T$ using the information encoded in $t_{i}$ about the likelihoods of formulas. For example, if $B_{i}^{\theta}\phi \in t_{i}$, this tells us that $\beta_{i}(t_{i})$ must assign probability at least $\theta$ to the subset of $X \times T$ where $\phi$ holds. Of course, in order to make sense of this, we must first define the event in $X \times T$ that corresponds to $\phi$.

As a first step toward this, we show that given a state-type tuple
$(x,t) \in X \times T$, the collection of formulas obtained by taking
the union of all these partial descriptions, namely $x \cup
\bigcup_{i} t_{i}$, is satisfiable. It is obvious that every
$0$-description $x \in X$ and $i$-description $t_{i} \in T_{i}$ is
individually satisfiable since, by definition, each is satisfied at
some world in $\M$. On the other hand, there is no guarantee that they
are all satisfied at the \textit{same} world in $\M$ (and in general
they may not be), so their joint satisfiability is not so obvious. 

\begin{lemma} \label{lem:sat}
For all $(x,t) \in X \times T$, the collection $x \cup \bigcup_{i}
t_{i}$ is satisfiable. 
\end{lemma}

\begin{proof}
  As observed, there are worlds $\omega_0, \ldots, \omega_n$ in $\M$ such
  that
$\omega_{0}$ satisfies $x$
and
$\omega_{i}$ satisfies $t_{i}$
for $i =1,\ldots, n$.  We now construct a model $\M^*$ and world $\omega^*$ in
$\M^*$ such that $\M^*$ consists of $n$ disjoint copies of $\M$
together with the world $\omega^*$; formally,
$\M^* = ((\Omega^*, (\PR_{i}^*)_{i \in I}), \pi^*)$, where 
\begin{itemize}
\item $\Omega^* =
\{(\omega,i): \omega \in \Omega, i \in \{1,\ldots, n\}\} \cup
\{\omega^*\}$;
\item $\pi^*(p) = \left\{\begin{array}{ll}
\cup_{i=1}^n (\pi(p) \times \{i\}) &\mbox{if $\omega_0 \notin \pi(p)$}\\
\cup_{i=1}^n (\pi(p) \times \{i\}) \cup \{\omega^*\} &\mbox{if $\omega_0 \in \pi(p)$}\\
   \end{array}\right.$
\item $\PR_{i}^*(\omega,i)(U \times \{i\}) =
  \PR_i(\omega)(U)$   for $\omega \in \Omega$, and
  $\PR_i(\omega^*)(U \times \{i\}) = 
  \PR_i(\omega_i)(U)$ 
 (so the support of 
   $\PR_{i}^{*}(\omega,i)$ and of $\PR_i(\omega*)$ is contained $\Omega
  \times \{i\}$). 
\end{itemize}
It is easy to check that $\omega^*$ agrees with $\omega_0$ on
propositional formulas and with $\omega_i$ on 
$i$-descriptions.  Thus, the desired result holds.
\end{proof}

In fact, not only is $x \cup \bigcup_{i} t_{i}$ satisfiable, but it determines a unique $\LBT$-description.

\begin{lemma} \label{lem:det}
There is a unique $\LBT$-description $D$ such that $D \supseteq x \cup \bigcup_{i} t_{i}$.
\end{lemma}

\begin{proof}
By Lemma \ref{lem:sat}, such a $D$ exists (take $D = D(\omega)$ for some $\omega$ that satisfies $x \cup \bigcup_{i} t_{i}$). Uniqueness follows from the following observation, easily proved by structural induction on $\phi$: for all $\phi \in \LBT$, either $x \cup \bigcup_{i} t_{i}$ entails $\phi$ or $x \cup \bigcup_{i} t_{i}$ entails $\lnot \phi$.
\end{proof}

Let $D(x,t)$ denote the unique description determined by $x \cup \bigcup_{i} t_{i}$ as in Lemma \ref{lem:det}. It is easy to see that $D(d_{0}(\omega), d(\omega)) = D(\omega)$. On the other hand, as mentioned above, the collection of descriptions of the form $D(x,t)$ may be strictly larger than those of the form $D(\omega)$, since some tuples $(x,t)$ may combine partial descriptions that are not simultaneously satisfied at any world in $\M$.

The description $D(x,t)$ provides a natural way to associate formulas with events in $X \times T$. For each $\phi \in \LBT$, define
$$\textstyle [\phi] \defeq \{(x,t) \in X \times T \: : \: \phi \in D(x,t)\}.$$

\begin{lemma} \label{lem:gen}
$\Sigma_{X \times T}$ is generated by the collection $\{[\phi] \: : \: \phi \in \LBT\}$.
\end{lemma}

\begin{proof}
It is easy to see that every $\phi \in \LBT$ is a Boolean combination of primitive propositions and formulas of the form $B_{i}^{\theta}\psi$; it follows that $\{[\phi] \: : \: \phi \in \LBT\}$ is the algebra generated by all sets of the form $[p]$ and $[B_{i}^{\theta}\psi]$. Now observe that $(x,t) \in [p]$ iff $p \in x$, so $[p] = E_{0}(p) \times T$, and similarly, $(x,t) \in [B_{i}^{\theta}\psi]$ iff $B_{i}^{\theta}\psi \in t_{i}$, so $[B_{i}^{\theta}\psi] = X \times E_{i}(B_{i}^{\theta}\psi) \times T_{-i}$. Thus, $\{[\phi] \: : \: \phi \in \LBT\} \subseteq \Sigma_{X \times T}$.

To see that $\Sigma_{X \times T}$ is in fact generated by this collection, it suffices to observe that if each of $E_{0}(\phi_{0})$, $E_{1}(\phi_{1})$, \ldots, $E_{n}(\phi_{n})$ is nonempty, then
$$E_{0}(\phi_{0}) \times E_{1}(\phi_{1}) \times \cdots \times E_{n}(\phi_{n}) = [\phi_{0} \land \phi_{1} \land \cdots \land \phi_{n}].$$
\end{proof}

We turn now to defining the probability measures
$\beta_{i}(t_{i})$. Each $t_{i} \in T_{i}$ is a collection of formulas
in $\LBT$ that bear on agent $i$'s beliefs. We can use these
formulas to constrain the space of possible outputs of
$\beta_{i}(t_{i})$. Moreover, provided $\LBT$ is rich enough, these
contraints yield a unique probability measure. 

Let $\P_{t_{i}}$ denote the set of all probability measures $\mu$ on $X \times T$ such that, for each $\phi \in \LBT$ and all $\theta \in \Theta$,
\begin{equation} \label{eqn:con}
\mu([\phi]) \geq \theta \; \dimp \; B_{i}^{\theta} \phi \in t_{i}.
\end{equation}

\begin{lemma} \label{lem:pti}
$\P_{t_{i}} \neq \emptyset$. Moreover, if $\Theta$ is dense in $[0,1]$, then $|\P_{t_{i}}| = 1$.
\end{lemma}

\begin{proof}
First we show that $\P_{t_{i}}$ is nonempty. Let $\omega$ be a world in $\M$ such that $d_{i}(\omega) = t_{i}$. For each $\phi \in \LBT$, define
$$\mu_{i,\omega}([\phi]) = \PR_{i}(\omega)(\valM{\phi}).$$
One can check that $\mu_{i,\omega}$ is a pre-measure on the algebra $\{[\phi] \: : \: \phi \in \LBT\}$ and satisfies (\ref{eqn:con}).
By Carath\'eodory's extension theorem \cite[Theorem 1.14]{Folland},
there is a unique extension $\tilde{\mu}_{i,\omega}$ of $\mu_{i,\omega}$ to the $\sigma$-algebra generated by $\{[\phi] \: : \: \phi \in \LBT\}$, which by Lemma \ref{lem:gen} is just $\Sigma_{X \times T}$. Therefore, by construction, $\tilde{\mu}_{i,\omega} \in \P_{t_{i}}$.

If $\Theta$ is dense in $[0,1]$, then it is easy to see that for all
$\phi \in \LBT$, if $\mu \in \P_{t_{i}}$ then 
$$\mu([\phi]) = \sup\{\theta \in \Theta \: : \: B_{i}^{\theta} \phi \in t_{i}\}.$$
It follows that $\P_{t_{i}} = \{\tilde{\mu}_{i,\omega}\}$.
\end{proof}


Let us restrict our attention for the time being to the case where $\Theta$ is a countable, dense subset of $[0,1]$; indeed, it is common to assume that $\Theta = [0,1] \cap \mathbb{Q}$. Countability ensures that $\LBT$ contains only countably-many modalities, and by Lemma \ref{lem:pti}, density allows us to define $\beta_{i}(t_{i})$ to be the unique element of $\P_{t_{i}}$. We then have the following:

\begin{proposition}
Let $\M$ be a probability model, and let $\TM \defeq (X, (T_{i})_{i
  \in I}, (\beta_{i})_{i \in I})$ as defined above. Then $\TM$ is a
type space. Define $\nu(p) \defeq E_{0}(p)$. Then $\IM \defeq
(\TM,\nu)$ is an interpreted type space, and for all $\phi \in
\LBT$, we have 
$$\omega \in \valM{\phi} \; \rimp \; (d_{0}(\omega), d(\omega)) \in \val{\phi}_{\IM}.$$
\end{proposition}

\begin{proof}
First we observe that $\Sigma_{\Delta(X \times T)}$ is generated by
all events of the form $\B^{\theta}([\phi])$; this follows from Lemma
\ref{lem:gen} together with \cite[Lemma 4.5]{HeSa98}.
Thus, to prove that $\beta_{i}$ is measurable it suffices
to prove that each set $\beta_{i}^{-1}(\B^{\theta}([\phi]))$ is
measurable. By definition, we know that $\beta_{i}(t_{i})([\phi]) \geq
\theta$ iff $B_{i}^{\theta} \phi \in t_{i}$; it follows that 
$$\beta_{i}^{-1}(\B^{\theta}([\phi])) = E_{i}(B_{i}^{\theta} \phi),$$
which is measurable by definition.
That $\beta_{i}(t_{i})$ concentrates on $t_{i}$ follows from the fact that
$$(x,t) \in [B_{i}^{\theta}\psi] \dimp B_{i}^{\theta}\psi \in t_{i} \dimp B_{i}^{1}B_{i}^{\theta}\psi \in t_{i} \dimp \beta_{i}(t_{i})([B_{i}^{\theta}\psi]) = 1.$$
Finally, the semantic equivalence follows by structural induction on $\phi$.
\end{proof}

\section{Universal Type Spaces and Canonical Models} \label{sec:uni}

\subsection{Universal type spaces}

The existence of a \emph{universal type space} \cite{MZ} underpins the use of
type spaces as a general framework for modeling beliefs: roughly
speaking, it guarantees that they do not rule out any possible
collection of beliefs. Individual type spaces, of course, can be quite
small and omit many configurations of beliefs. The universal type
space, by contrast, essentially includes all possible configurations
of belief; in particular, this means we need not be concerned with
gaps in our representation of games.

Formally, given type spaces $\T = (X, (T_{i})_{i \in I},
(\beta_{i})_{i \in I})$ and $\T' = (X, (T_{i}')_{i \in I},
(\beta'_{i})_{i \in I})$ (with a common set $X$ of states), a profile
of functions $f_{i}: T_{i} \to T_{i}'$ constitutes a \defin{type morphism} from $\T$ to $\T'$ provided that, for each $i \in I$, $t_{i}
\in T_{i}$, and each event $E \subseteq X \times T'$, 
$$\beta_{i}'(f_{i}(t_{i}))(E) = \beta_{i}(t_{i})(f^{-1}(E)),$$
where $f: X \times T \to X \times T'$ is defined by $f = (id_{X},
f_{1}, \ldots, f_{n})$. Roughly speaking, this says that each $f_{i}$
assigns to each $t_{i} \in T_{i}$ a type $f_{i}(t_{i}) \in T_{i}'$
that agrees with $t_{i}$ on the probabilities of all events, where
events in $\T$ and $\T'$ are identified via the correspondence given
by $f$. A type space $\T^{*}$ is called \defin{universal (for $X$)}
if, for every type space $\T = (X, (T_{i})_{i \in I}, (\beta_{i})_{i
  \in I})$, there exists a unique type morphism from $\T$ to
$\T^{*}$. Thus, each such $\T$ can be thought of as existing
``inside'' $\T^{*}$ (via the mapping $f$). 

Type morphisms are defined so as to preserve the structure of belief. Indeed, given any interpretation $\nu: \Phi \to \Sigma_{X}$, it is easy to see that if $(f_{1}, \ldots, f_{n})$ is a type morphism from $\T$ to $\T'$, then for any $(x,t) \in X \times T$ and any $\phi \in \LBI(\Phi)$, we have
$$(x,t) \in \val{\phi}_{(\T,\nu)} \dimp f(x,t) \in \val{\phi}_{(\T',\nu)}.$$
As a consequence, the universal type space for $X$ satisfies all the $\LBI(\Phi)$-descriptions that are satisfied in \textit{some} type space over $X$. It is natural to wonder whether this property characterizes the universal type space; the connection with canonical models we now present essentially amounts to a formalization of this idea.

\subsection{Canonical models} \label{sec:can}

The classical canonical model construction is used to prove completeness of various modal systems. Given some axiom system $\mathsf{AX}$ of interest, a model is constructed wherein each world corresponds to a maximal $\mathsf{AX}$-consistent set of formulas, with additional structure derived from the properties of these sets of formulas.

The construction we present here differs in that we are not concerned
with axiomatics---indeed, for logics that fail to be compact (such as,
notably, the logic of $\LBI$ as interpreted in probability frames),
consistent sets of formulas need not be satisfiable, so the canonical
model construction fails. Nonetheless, we can adapt this construction
by replacing ``consistent'' with ``satisfiable''; in other words, we
can build a model in which the worlds are exactly the
$\LBI$-descriptions.%
\footnote{As we said above, a similar construction appears in \cite{Aumann99a}, though
  the connection to type spaces is not explored in any depth.} 
Intuitively, such a model contains a world satisfying every such
description; ultimately, we will show that we can obtain a universal
type space by constructing such a model and then translating it into a
type space as in Section \ref{sec:fac}. 

Consider a fixed language $\LBI(\Phi)$ and a class of probability models $\mathscr{C}$; let $\bar{\Omega}$ denote the set of all $\LBI(\Phi)$-descriptions satisfiable in some model in $\mathscr{C}$. Define $\hat{\phi} = \{\bar{\omega} \in \bar{\Omega} \: : \: \phi \in \bar{\omega}\}$, and let $\Sigma_{\bar{\Omega}}$ be the $\sigma$-algebra generated by the collection $\A = \{\hat{\phi} \: : \: \phi \in \LBI(\Phi)\}$. Define $\mu_{i, \bar{\omega}}: \A \to [0,1]$ by
$$\mu_{i, \bar{\omega}}(\hat{\phi}) = \sup\{\theta \in [0,1] \: : \: B_{i}^{\theta}\phi \in \bar{\omega}\}.$$
It is not hard to check that $\mu_{i,\bar{\omega}}$ is a pre-measure
on the algebra $\A$,
so, by Carath\'eodory's extension theorem, it can be
extended to a unique 
probability measure on $\Sigma_{\bar{\Omega}}$; let
$\bar{\PR}_{i}(\bar{\omega})$ denote this extension. Finally, for each
$p \in \Phi$, set $\bar{\pi}(p) = \hat{p}$. 

\begin{proposition}
  $\bar{\M} = (\bar{\Omega}, (\bar{\PR}_{i})_{i \in I}), \bar{\pi})$ is a
    probability model, and for all $\phi \in \LBI$, we have
  $\val{\phi}_{\bar{\M}} = \hat{\phi}$. Moreover, $\bar{\M}$ is
  universal for $\mathscr{C}$ in the sense that, for all $\M \in
  \mathscr{C}$, there is a truth-preserving map (namely, $D$, the
  description map for $\M$) from $\M$ to $\bar{\M}$. 
\end{proposition}

Call $\bar{\M}$ the \defin{universal probability model for $\mathscr{C}$ over $\LBI(\Phi)$}.
As we mentioned earlier, Meier \citeyear{Meier12} works with an
infinitary version of the language $\LBI(\Phi)$ and constructs a
canonical model  for that language.  Call his language $\LBII$.
Although $\LBII$ is infinitary, as observed in \cite[Lemma
  4.1]{HP08a}, every $\LBI$-description can be uniquely extended to an
$\LBII$-description.  It follows that the canonical model for the
language $\LBI$ is isomorphic to the canonical model for $\LBII$.
Meier shows that the canonical model for $\LBII$ is universal.  Of
course, it follows that the canonical model for $\LBI$ is also
universal.  We given an independent proof of this result here, since
it allows us to connect universal type spaces to the language
considerations discussed earlier.

\subsection{Translation} \label{sec:trs}

Let $X$ be a measurable space of states where $\Sigma_{X}$ is
generated by the singletons $\{x\}$.%
\footnote{It is possible to weaken this condition to the following:
  for every $x,y \in X$, there exists a ``separating event'' $E \in
  \Sigma_{X}$ such that $x \in E$ and $y \notin E$. The issue here is
  that if $X$ contains points that are not separated in this way, they
  will not differ on any description and so the universal model
  construction we employ below will end up identifying them. Notice,
  however, that this is only a problem because the universal type space
for state space $X$ is required to use $X$ as the state space, even
when $X$ contains 
  ``redundant'' states that are not separated by any
  event. Intuitively, however, this is unnecessary---a slightly
  relaxed notion of a universal type space would simply require that
  its state space be rich enough to reflect the measure structure of
  $X$, rather than its set-theoretic structure. And indeed, this is
  essentially what you get by running the construction below without
  the separability requirement articulated above.} 
We construct a universal type
space for $X$ by first constructing a universal model as in Section
\ref{sec:can}. Consider the language $\LBI(X)$ (i.e., where $\Phi =
X$) and the class $\mathscr{C}_{X}$ of probability models such that
$\{\pi(x) \: : \: x \in X\}$ partitions $\Omega$. Intuitively, this
condition hard-codes the constraint that exactly one state $x \in X$
is the ``true'' state of the world. 

\begin{theorem} \label{thm:can}
Let $\bar{\M}$ be the universal probability model for $\mathscr{C}_{X}$ over $\LBI(X)$. Then the type space $\T_{\bar{\M}}$ is universal for $X$.
\end{theorem}

\begin{proof}
The state space for $\T_{\bar{\M}}$ is, by definition, the collection $\{d_{0}(\bar{\omega}) \: : \: \bar{\omega} \in \bar{\Omega}\}$; it is easy to see that each set $d_{0}(\bar{\omega})$ contains exactly one element of $X$, and this correspondence is
a measurable bijection with measurable inverse.
So $\T_{\bar{\M}}$ has the ``right'' state space.

Next, let $\T = (X, (T_{i})_{i \in I}, (\beta_{i})_{i \in I})$ be any type space based on $X$. We must produce a (unique) type morphism from $\T$ to $\T_{\bar{\M}}$. To do so, define $\nu: \Phi \to \Sigma_{X}$ by $\nu(x) = \{x\}$, let $\I = (\T,\nu)$ be the corresponding interpreted type space, and consider the model $\M_{\I}$ obtained from $\I$ as in Proposition \ref{pro:t-m}. It is easy to see that $\M_{\I} \in \mathscr{C}_{X}$, and because of this, for each $(x,t) \in X \times T$ and $i \in I$, there is a unique $d_{i}(\bar{\omega})$ that is satisfied at $(x,t)$. In this case, define $f_{i}(t_{i}) = d_{i}(\bar{\omega})$.
\end{proof}

Theorem \ref{thm:can} realizes the intuition that the universal type
space for $X$ is precisely the type space that satisfies all and only the
$\LBI(\Phi)$-descriptions that are satisfied in some type space over
$X$. Thinking of universal type spaces in this way makes the
dependence on language plain, and suggests alternative notions of
``universal type spaces'' obtained by varying the language over which
the universal quantification takes place. That is, given a class of
type spaces $\mathscr{T}$ and a language $\mathcal{L}$ interpretable
in those type spaces in $\mathscr{T}$, we can define a type space
$\T^{*}$ to be \defin{universal for $\mathscr{T}$ with respect to
  $\mathcal{L}$} provided every $\mathcal{L}$-description satisfiable
in $\mathscr{T}$ is (uniquely) satisfied in $\T^{*}$. Naturally, we
might hope to construct $\T^{*}$ by transforming an appropriate
canonical/universal model. The translation defined in Section
\ref{sec:tra} does the job for languages of the form $\LBT$ when
$\Theta$ is dense in $[0,1]$. Generalizing this result to other
languages, both richer and poorer, is the subject of ongoing
research.

One natural way to coarsen the language is by dropping the assumption
that $\Theta$ is dense in $[0,1]$. An extreme case of this would be to
take $\Theta = \{1\}$, corresponding to a standard modal language of
qualitative, ``probability $1$'' belief (see, e.g., \cite{Hal31}). In
this case, the sets of measures $\P_{t_{i}}$ defined in Section
\ref{sec:tra} encode only information regarding those events that
$t_{i}$ assigns probability 1 to.
Another natural modification to the language is to enrich it with a
knowledge modality. Logics of knowledge and belief have been
well-studied, 
and canonical models certainly exist in such settings (see \cite{Len}
and the references in \cite[Chapter 8]{Hal31}). By contrast,
\emph{knowledge spaces}, an epistemic analogue to type spaces, have
been shown \textit{not} to permit a universal object \cite{FHV1,HeSa1}. 
What is the source of this mismatch? Does the translation technique we
present fundamentally fail to generalize to models of knowledge? Or
can the canonical model construction in the modal case inform a new,
type-theoretic representation of knowledge that does enjoy a universal
model? We leave these questions to future work.

\section{Conclusion} \label{sec:con}
We have related probability frames and type spaces in a way that makes
clear the critical role of language. Our approach allows us to show 
the deep connections between the canonical models that are
standard in the modal logic community and the universal type spaces
that play a critical role in epistemic game theory. We believe that
further work, considering different choices of language, will further
illuminate the connections between
these two modeling paradigms.

\bibliographystyle{eptcs}
\bibliography{z,joe}

\end{document}